\documentclass[twocolumn,pra,superscriptaddress]{revtex4-2}
\usepackage{graphicx,subfigure}          % For the figures and subfigures
\usepackage{amsfonts,amsmath,amssymb, amsthm}    % For various math symbols
\usepackage{dsfont}                      % For the identity operator
\usepackage[dvipsnames]{xcolor}          % To use colors
\usepackage[
colorlinks=true,
linkcolor=blue,
citecolor=blue,
urlcolor=blue]{hyperref}                % For the use of hyperlinks
\usepackage{comment}                    % Comment package for MacBook
\usepackage{physics}                    % To use Dirac Notation
\usepackage{bm}
\usepackage{mdframed}
\usepackage{enumitem}
\usepackage{soul}

\newtheorem{theorem}{Theorem}

\newtheorem{definition}{Remark}

\newcommand{\IHPC}{A*STAR Quantum Innovation Centre (Q.InC), Institute of High Performance Computing (IHPC), Agency for Science, Technology and Research (A*STAR), 1 Fusionopolis Way, \#16-16 Connexis, Singapore, 138632, Republic of Singapore.\looseness=-1}

\newcommand{\CQuERE}{Centre for Quantum Engineering, Research and Education, TCG CREST, Sector V, Salt Lake, Kolkata 700091, India.\looseness=-1}

\newcommand{\sutd}{Science, Mathematics and Technology Cluster, Singapore University of Technology and Design, 8 Somapah Road, Singapore 487372, Singapore}

\begin{document}

%%%%%%%%%%%%%%%%%%%%%%%%%%%%%%%%%%%%%%%%%%%%%%%%%%%%%%%%%%%%%%%%%%%%%%%%%%%%%%%%%%%%%%%%

\title{Local contextuality-based self-tests are sufficient for randomness expansion secure against quantum adversaries}

%%%%%%%%%%%%%%%%%%%%%%%%%%%%%%%%%%%%%%%%%%%%%%%%%%%%%%%%%%%%%%%%%%%%%%%%%%%%%%%%%%%%%%%%

\author{Jaskaran Singh}
\email{jaskaran@gs.ncku.edu.tw}
\affiliation{Department of Physics and Center for Quantum Frontiers of Research \&
Technology (QFort), National Cheng Kung University, Tainan 701, Taiwan}

\author{Cameron Foreman}
\email{cameron.foreman@quantinuum.com}
\affiliation{Quantinuum, Partnership House, Carlisle Place, London SW1P 1BX, United Kingdom}
\affiliation{Department of Computer Science, University College London, London, United Kingdom}

\author{Kishor Bharti}
\email{kishor.bharti1@gmail.com}
\affiliation{\IHPC}
\affiliation{\CQuERE}
\affiliation{\sutd}

\author{Ad\'an Cabello}
\email{adan@us.es}
\affiliation{Departamento de F\'{\i}sica Aplicada II, Universidad de Sevilla, 41012 Sevilla, Spain}
\affiliation{Instituto Carlos~I de F\'{\i}sica Te\'orica y Computacional, Universidad de Sevilla, 41012 Sevilla, Spain}

%%%%%%%%%%%%%%%%%%%%%%%%%%%%%%%%%%%%%%%%%%%%%%%%%%%%%%%%%%%%%%%%%%%%%%%%%%%%%%%%%%%%%%%%

\begin{abstract}
In quantum cryptography, secure randomness expansion involves using a short private string of random bits to generate a longer one, even in the presence of an adversary who may have access to quantum resources. In this work, we demonstrate that local contextuality-based self-tests are sufficient to construct a randomness expansion protocol that is secure against computationally unbounded quantum adversaries. Our protocol is based on self-testing from non-contextuality inequalities and we prove that our scheme asymptotically produces secure random numbers which are $\mathcal{O}(m\sqrt{\epsilon})$-close to uniformly distributed and private, where $\epsilon$ is the robustness parameter of the self-test and $m$ is the length of the generated random bit string. Our protocol is semi-device-independent in the sense that it inherits any assumptions necessary for the underlying self-test. 
\end{abstract}

%%%%%%%%%%%%%%%%%%%%%%%%%%%%%%%%%%%%%%%%%%%%%%%%%%%%%%%%%%%%%%%%%%%%%%%%%%%%%%%%%%%%%%%%

\maketitle

%%%%%%%%%%%%%%%%%%%%%%%%%%%%%%%%%%%%%%%%%%%%%%%%%%%%%%%%%%%%%%%%%%%%%%%%%%%%%%%%%%%%%%%%

\textit{Introduction.---}Random bits are crucial for numerous applications, including cryptography and simulations. Randomness expansion is the most common method of generating random bits, whereby a short string of random bits is used to securely generate a longer one. Quantum theory has proven to be highly effective in performing randomness expansion (RE)~\cite{AM16, HG17, MYC16, BKG18, ROT94, DYS08}. Device-independent quantum randomness expansion (DI-QRE) protocols, which leverage quantum theory, offer the highest level of security and several are even secure in the presence of computationally unbounded quantum adversaries~\cite{C11, CK11, BRC20, PAM10, PM13, FGS13, NBS18}. 
These protocols do not require detailed knowledge of the devices used; instead, they rely on the violation of Bell inequalities~\cite{B64, B66, XSS23} to certify secure randomness. However, DI-QRE protocols have two strict experimental requirements that are challenging to implement in practice: (i) quantum devices need to be isolated from each other and (ii) a detection loophole-free Bell inequality test.

To address these challenges, semi-device-independent quantum randomness expansion (SDI-QRE) protocols have been developed, which relax some of the strict experimental requirements of DI-QRE at the cost of requiring some additional knowledge about the devices used. These protocols offer a valuable trade-off between resource constraints and security guarantees.
Numerous SDI-QRE protocols have been developed, for example, those that constrain the experimental devices behavior by contextuality, steering, Leggett-Garg inequalities, energy constraints or quantum interference effects~\cite{CZM15,NGZ16, BKB17,XSW16, HWCGP17, ATVV21, WBGB22,SSC17, BME17, TZASVV21}.

However, given the increasing capabilities of quantum technologies, it is essential to develop protocols that are secure against quantum adversaries. To this end, only a small number of SDI-QRE schemes with this level of security have been developed. Most of these schemes require that the measurements devices are fully characterized and trusted~\cite{CZY16, AMV18, DWH20}, while another scheme~\cite{S23} requires projective measurements and devices which cannot be classically pre-programmed to yield a set of outcomes.

Our work addresses this issue by proposing a SDI-QRE protocol based on self-testing of contextual correlations that generates universally-composable \cite{C01} random bits secure against a computationally unbounded quantum adversary, without fully characterizing the measurement and preparation devices.
Similar to Bell inequalities, there are non-contextuality (NC) inequalities whose violation serves as a signature of non-classical (contextual) correlations. However, in contrast to Bell inequalities, violation of NC inequalities can be observed on single local systems.
In some cases, whenever contextual correlations are observed that are $\epsilon$-close to the quantum maximum of a NC inequality, one can robustly self-test the underlying quantum state and measurements~\cite{BRV19}. Using this fact, we construct a protocol for SDI-QRE based on self-testing via NC inequalities and prove that our protocol can produce $m$ random bits that are $\mathcal{O}(m\sqrt{\epsilon})$-close (in trace norm) to being uniformly distributed and uncorrelated from a quantum adversary.

Our results extend RE via self-testing, which has been a unique feature of DI scenarios, to the SDI scenario. This is significant because self-testing results in the context of proving the security of DI-QRE are not typically advantageous, since there are several other well-established techniques, such as the entropy accumulation theorem~\cite{DFR20, MFS22} or quantum probability estimation~\cite{XFK20}. In our case, this relation enables us to establish security of the generated random bit string without the need for these techniques, whose applicability in the SDI setting is not always straightforward.
Finally, recent loophole-free experimental implementations of NC inequalities~\cite{WZL22, HXA23} have achieved close to maximal quantum violation of the NC inequalities while also enforcing the required assumptions, demonstrating that our scheme is already experimentally accessible.

%%%%%%%%%%%%%%%%%%%%%%%%%%%%%%%%%%%%%%%%%%%%%%%%%%%%%%%%%%%%%%%%%%%%%%%%%%%%%%%%%%%%%%%%%%%%%%%%%%%%%%%%%%%%%%%

%\subsection{Self-tests of non-contextuality inequalities}
%\label{sec:nc_2}

%%%%%%%%%%%%%%%%%%%%%%%%%%%%%%%%%%%%%%%%%%%%%%%%%%%%%%%%%%%%%%%%%%%%%%%%%%%%%%%%%%%%%%%%%%%%%%%%%%%%%%%%%%%%%%%

\textit{Self-testing from non-contextuality inequalities.---} 

Contextuality (and, in particular, contextuality for ideal, i.e., projective, measurements) is a non-classical feature of quantum theory~\cite{KS67,KCB08,budroni2022kochen} that has found practical applications~\cite{BRX22,xu2024certifying,UZZ13, UZZ20, SBA17,ray2021graph,yu2022quantum,bharti2023power,arora2024computational}, including self-testing~\cite{BRV19,BRX22,saha2020sum,bharti2019local}. In this work, we build on the result of Ref.~\cite{BRV19}, showing that states and measurements of individual quantum systems can be self-tested using certain NC inequalities.

Consider a finite set of measurements represented in quantum theory by the projectors $\{\Pi_i\}_{i=1}^N$, each of them with possible outcomes $a_i\in \lbrace 0,1 \rbrace$. Let $\mathcal{G} = (\mathcal{V}, \mathcal{E})$ be the graph of compatibility of $\{\Pi_i\}_{i=1}^N$, where each of the elements of the set of vertices $\mathcal{V} = \{v_i\}_{i=1}^{N}$ represents one element of $\{\Pi_i\}_{i=1}^N$, and there is an edge $(v_i, v_j) \in \mathcal{E} \subseteq \mathcal{V} \times \mathcal{V}$ if and only if $\Pi_i$ and $\Pi_j$ are jointly measurable (compatible). If $\mathcal{G}=C_N$ with $N \ge 5$ and $N$ odd (i.e., a cycle with $N$ vertices), then the following NC inequality is satisfied by any non-contextual hidden-variable (NCHV) theory \cite{AQB13}:
\begin{equation}
    \mathcal{\beta} = \sum_{i = 1}^N p(1,0|i,i+1) \leq \beta_{\text{NCHV}},
    \label{eq:nc_inequality}
\end{equation}
where $p(1,0|i,i+1)$ is the probability of obtaining outcomes $1$ and $0$ for $\Pi_i$ and $\Pi_{i+1}$ respectively, the sum in $i+1$ is taken such that $N + 1 = 1$
and $\beta_{\text{NCHV}}$ is the maximum of $\beta$ for NCHV theories. In this case, $\beta_{\text{NCHV}} = \alpha(C_N)$, which is the independence number of $C_N$~\cite{CSW14}, while
in quantum theory the maximum of $\beta$ is $\beta_{\text{QT}} = \vartheta(C_N)$, where $\vartheta(C_N)$ is the Lov{\'a}sz number of $C_N$~\cite{CSW14}. See Appendix~\ref{app:selftest} for more details.

An optimal quantum strategy $\mathcal{S} = (\rho, \{ \Pi_i \}_{i=1}^{N})$ consists of a quantum state $\rho$ and a set of projective two-outcome measurements $\{ \Pi_i \}_{i=1}^{N}$ which achieves $\beta = \beta_{\text{QT}}$. Hereafter, we denote by $\ket{a_i}\!\bra{a_i}$ the post-measurement state after an ideal measurement of $\Pi_i$ with outcome $a_i$.

\begin{definition}[Robust self-testing from odd-cycle NC inequalities~\cite{BRV19}]
\label{def:self_test}
    If $\mathcal{G} = (\mathcal{V}, \mathcal{E})$ is a cycle with $N \ge 5$ vertices and $N$ odd, and a quantum strategy $\mathcal{S}=\left(\ket{v_0}\!\bra{v_0}, \lbrace \Pi_i = \ket{v_i}\!\bra{v_i}\rbrace_{i = 1}^{N}\right)$, where $\ket{v_0}\!\bra{v_0}$ is the initial state, $\lbrace v_i\rbrace_{i = 1}^{N} \in \mathcal{V}$ and  $\text{tr}\left(\Pi_i \Pi_j\right) = 0$ for all $(v_i, v_j)\in \mathcal{E}$, achieves $\beta = \beta_{\text{QT}}$, then the NC inequality Eq.~\eqref{eq:nc_inequality} serves as a robust local self-test for this strategy in the sense that for any other quantum strategy $\tilde{\mathcal{S}}= \left(\ket{\tilde{v}_0}\!\bra{\tilde{v}_0}, \lbrace \tilde{\Pi}_i = \ket{\tilde{v}_i}\!\bra{\tilde{v}_i}\rbrace_{i = 1}^{N}\right)$ that achieves $\beta = \beta_{\text{QT}} - \epsilon$, there exists an isometry $V_A$ such that $\norm{V_A\ket{\tilde{v}_i}\!\bra{\tilde{v}_i}V^\dagger_A - \ket{v_i}\!\bra{v_i}} \leq \mathcal{O}(\sqrt{\epsilon})$ for $i = 0, 1, \ldots, N$ and $V^\dagger_A V_A = \mathds{1}$.
\end{definition} 

Here, $\norm{A} = \text{tr}\left(\sqrt{A^\dagger A}\right)$ denotes the trace norm of a matrix $A$. We note that a non-robust version of Remark~\ref{def:self_test}, in which $\epsilon = 0$, also holds for the case of even $N$-cycle NC inequalities~\cite{BRX22}.
Next, we introduce our QRE protocol. 

%%%%%%%%%%%%%%%%%%%%%%%%%%%%%%%%%%%%%%%%%%%%%%%%%%%%%%%%%%%%%%%%%%%%%%%%%%%%%%%%%%%%%%%%%%%%%%%%%%%%%%%%%%%%%%%

%\subsection{Protocol}
%\label{subsec:protocol}

%%%%%%%%%%%%%%%%%%%%%%%%%%%%%%%%%%%%%%%%%%%%%%%%%%%%%%%%%%%%%%%%%%%%%%%%%%%%%%%%%%%%%%%%%%%%%%%%%%%%%%%%%%%%%%%
% Fig.
%%%%%%%%%%%%%%%%%%%%%%%%%%%%%%%%%%%%%%%%%%%%%%%%%%%%%%%%%%%%%%%%%%%%%%%%%%%%%%%%%%%%%%%%

\begin{figure}
\begin{mdframed}
{\bf Protocol for QRE - Odd $N$-cycle}
\\
\hrule
\vspace{0.3cm}
\textbf{Parameters and notation:}
\begin{itemize}[leftmargin = 0.3cm]
    \item[] $n \in \mathbb{N}$ - Total number of rounds.  

    \item[] $N\geq 5\in \mathbb{N}/2\mathbb{N}$ - Total number of measurements, which is odd and at-least $5$.

    \item[] $\Pi_i$, $i\in \lbrace 1, \ldots, N\rbrace$ - Two-outcome projective measurements with outcomes labeled by $a_i \in \lbrace 0, 1\rbrace$ corresponding to $\mathds{1} - \Pi_i$ and $\Pi_i$ respectively.

    \item[] $\rho$ - Initial quantum state on which the measurements are performed.

    \item[] $\beta$ - Observed value of $N$-cycle NC inequality, as defined in Eq.~\eqref{eq:nc_inequality}, with maximum quantum value $\beta_{\text{QT}}$. 
    
    \item[] $\left(\rho, \lbrace \Pi_i\rbrace_{i = 1}^{N}\right)$ - Quantum realization that obtains $\beta = \beta_{\text{QT}}$.

    \item[] $\epsilon \in \left( 0, 1 \right)$ - Parameter chosen to quantify a permissible value of deviation of $\beta$ from $\beta_{\text{QT}}$.

    \item[] $\omega_0 = \cos\frac{\pi}{N}$ and $\omega_1 = 1$. 

    \item[] $q \in \left[0,1\right]$ and $T, l \in \lbrace 0,1\rbrace$.
\end{itemize}
\hrule
\vspace{0.5mm}
\hrule
\vspace{0.3cm}
\textbf{Procedure}
\begin{enumerate}[leftmargin = 0.3cm]
    \item Alice chooses a quantum strategy $\left(\rho, \lbrace \Pi_i\rbrace_{i = 1}^{N}\right)$ which achieves $\beta_{\text{QT}}$ of an odd $N$-cycle NC inequality and sets $j = 1$.
    \item \textbf{While} $j \leq n$:
    \begin{itemize}
        \item[] For $q \in \left[0, 1\right]$, choose $T_j = 0$ with probability $1 - q$ and $T_j = 1$ otherwise. 

        \item[] \textbf{If} $T_j = 0$: (Generation round) 
        
        \begin{itemize}
            \item[] Perform the measurement $\Pi_1$ on state $\rho$ to obtain outcome $a_1$.

            \item[] \textbf{If} $a_1 = 0$:
            \begin{itemize}
                \item[] Record $a_1$ as $k_j$ with probability $\omega_0$.
            \end{itemize}

            \item[] \textbf{Else} $a_1 = 1$:
            \begin{itemize}
                \item[] Record $a_1$ as $k_j$ with probability $\omega_1$.
            \end{itemize}
        \end{itemize}

        \item[] \textbf{Else}: (Spot-check round)
        \begin{itemize}
            \item[] Randomly choose $i \in \lbrace 1, \ldots, N\rbrace$ and $l \in \{0, 1\}$ with uniform probability and compute $l' = (i + (-1)^l \mod{N}) + 1$.
        
            \item[] Perform the measurement $\Pi_{i}$ on state $\rho$ to obtain outcome $a_i$. Perform the measurement $\Pi_{l'}$ on the post-measurement state to obtain outcome $a_{l'}$. Record $a_i, a_{l'}, i, l'$.
        \end{itemize} 
        
        \item[] Set $j = j+1$.
    \end{itemize}

    \item Using the statistics from all spot-check rounds evaluate $\beta$. 

    \item[] \textbf{If} $\beta_{\text{QT}} - \beta \geq \epsilon$:
    \begin{itemize}
        \item[] Abort the protocol.
    \end{itemize}

    \item[] \textbf{Else}: 
    \begin{itemize}
        \item[] Obtain the \textit{final random bit string} $\bm{k}$ as a concatenation of all bit values $k_j$.
    \end{itemize}
\end{enumerate}
\end{mdframed}
\caption{The protocol for RE using self-tests of contextual correlations in odd $N$-cycle scenarios.}
\label{fig:protocol}
\end{figure}

%%%%%%%%%%%%%%%%%%%%%%%%%%%%%%%%%%%%%%%%%%%%%%%%%%%%%%%%%%%%%%%%%%%%%%%%%%%%%%%%%%%%%%%%

\textit{Protocol.---}We consider a single party, Alice, who aims to expand her private randomness (generate more randomness than she consumes). The idea of the protocol is that each round is randomly chosen by Alice (with bias $q$) to be one of two different classes of experimental rounds: \textit{spot-check rounds}, where Alice checks for a violation of a NC inequality~\eqref{eq:nc_inequality}, and \textit{generation rounds}, where the outcomes are selectively recorded, a process termed as \textit{post-selection}, and used to generate the final random bit string, $\bm{k}$, of length $m$. The step of post-selection is necessary in our protocol to ensure that the generated bit string is uniform.
We assume that Alice's measurements are repeatable, satisfy no-disturbance, and her devices are memoryless, which are necessary conditions for self-testing via NC inequalities. Additionally, Alice is assumed to have a trusted classical computer for performing computations, the correlations estimated during spot-checking rounds are identical to those in the generation rounds and all sources of errors are assumed to impact all experimental rounds equally. Further details on these assumptions are provided in the Appendix.

We present our QRE protocol in Fig.~\ref{fig:protocol} for NC inequalities corresponding to a graph $\mathcal{G}$ with an odd number of vertices, connected pairwise in a cycle (known as odd $N$-cycle NC inequalities; see Ref.~\cite{AQB13} for more details). However, our results are general and apply to arbitrary NC inequalities that certify a self-test, equivalent to that in Refs.~\cite{BRV19, BRX22}. We demonstrate additional examples by constructing an equivalent protocol using (non-robust) self-testing from even $N$-cycle NC inequalities in the Appendix. 

%%%%%%%%%%%%%%%%%%%%%%%%%%%%%%%%%%%%%%%%%%%%%%%%%%%%%%%%%%%%%%%%%%%%%%%%%%%%%%%%%%%%%%%%%%%%%%%%%%%%%%%%%%%%%%%

%\subsection{Security of the protocol}
%\label{subsec:security}

%%%%%%%%%%%%%%%%%%%%%%%%%%%%%%%%%%%%%%%%%%%%%%%%%%%%%%%%%%%%%%%%%%%%%%%%%%%%%%%%%%%%%%%%%%%%%%%%%%%%%%%%%%%%%%%

\textit{Security.---}In order to prove security we consider that a computationally unbounded quantum adversary, Eve, may share some correlations with Alice. Therefore, the joint state of Alice's post-measurement state corresponding to the final random bit string $\bm{k}$, of length $m$, and Eve can be written as a classical-quantum (cq) state 
\begin{equation}
    \rho_{KE} = \sum_{\bm{k}\in \lbrace 0, 1\rbrace^m} p(\bm{k}) \ket{\bm{k}}\!\bra{\bm{k}} \otimes \rho^{\bm{k}}_E,
\end{equation}
where $\rho^{\bm{k}}_E$ is the reduced state of Eve conditioned on $\bm{k}$. 

The bit string $\bm{k}$ is said to be $\epsilon_{sec}$-secure if Eve can distinguish it from a uniform distribution with probability at most $1/2 + \epsilon_{sec}/2$. For cq-states, this distinguishability is quantified by a bound on the trace norm:
% (due to the maximum guessing probability in quantum state discrimination)
\begin{equation}
    \norm{\rho_{KE} - 2^{-m}\mathds{1}_{K} \otimes \rho_E} \leq \epsilon_{sec},
    \label{eq:security}
\end{equation}
where $\mathds{1}_{K}$ denotes the $m \times m$ identity matrix and $\rho_E$ the reduced state of Eve, which is notably uncorrelated with the state of Alice. 

To show that Alice's final random bit string is $\epsilon_{sec}$-secure, we first demonstrate that Alice's post-measurement state, in each round, is uncorrelated with Eve when Alice uses a quantum strategy achieving $\beta = \beta_{\text{QT}}$ in the odd $N$-cycle NC inequality. We then show a robust version of this statement: Alice's post-measurement state, in each round, is $\mathcal{O}(\sqrt{\epsilon})$-close to being uncorrelated with Eve if her quantum strategy achieves $\beta \geq \beta_{\text{QT}} - \epsilon$. Finally, we prove that after Alice performs post-selection in the generation rounds, the $m$ outcomes in her final random bit string satisfies Eq.~\eqref{eq:security} with $\epsilon_{sec} = \mathcal{O}(m\sqrt{\epsilon})$. We defer the proofs of all the theorems to the Appendix.

\begin{theorem}
\label{thm:theorem1}
    For any quantum strategy $\tilde{\mathcal{S}} = \left(\tilde{\rho}_A, \lbrace \tilde{\Pi}_i = \ket{\tilde{v}_i}\!\bra{\tilde{v}_i}_A\rbrace_{i = 1}^{N}\right)$ implemented by Alice that achieves $\beta = \beta_{\text{QT}}$, the adversary Eve is uncorrelated with her post-measurement state, which can be written as
    \begin{equation}
    \rho_{AE} = \sum_{a_i = 0}^{1} p(a_i|i) \ket{a_i}\!\bra{a_i}_A \otimes \ket{\xi}\!\bra{\xi}_E,
\end{equation}
    where $\rho_{AE}$ denotes the joint state of Alice and Eve and $\ket{\xi}_E$ denotes the system held by Eve. 
\end{theorem} 

Next, we prove our results for the case when Alice achieves a violation of the odd NC inequality which is only $\epsilon$-close to the maximal one.

\begin{theorem}
    \label{thm:theorem2}
    For any quantum strategy $\tilde{\mathcal{S}} = \left(\tilde{\rho}_A, \lbrace \tilde{\Pi}_i = \ket{\tilde{v}_i}\!\bra{\tilde{v}_i}_A\rbrace_{i = 1}^{N}\right)$ implemented by Alice that achieves $\beta = \beta_{\text{QT}} - \epsilon$ $(\epsilon \ll 1)$, the post-measurement state of Alice and Eve, $\rho_{AE}$, is $\mathcal{O}( \sqrt{\epsilon})$-close (in trace distance) to a state which is uncorrelated to Eve, i.e.,
    \begin{equation}
        \norm{\rho_{AE} - \sum_{a = 0}^{1} p(a_i|i) \ket{a_i}\!\bra{a_i}_A \otimes \ket{\xi}\!\bra{\xi}_E} \leq \mathcal{O}(\sqrt{\epsilon}),
    \label{eq:aforementioned}
    \end{equation}
    where $\ket{\xi}_E$ denotes the system held by Eve.
    \end{theorem}

Currently, we have only shown that each of Alice's post-measurement state in the spot-check rounds is $\mathcal{O}(\sqrt{\epsilon})$-close to being uncorrelated with Eve. Next, we translate Theorem~\ref{thm:theorem2}, which is a result on Alice's spot-check rounds, to Alice's generation rounds. Note that Alice performs an additional step of post-selection in the generation rounds. Therefore, for these rounds the probability distribution of the outcomes $p(a_i|i)$ must be updated to the probability distribution after post-selection denoted by $\hat{p}(a_i|i)$ and given as
\begin{equation}
    \hat{p}(a_i|i) = \frac{\omega_{a_i} p(a_i|i)}{\sum_{a_i = 0}^{1} \omega_{a_i} p(a_i|i)} = \frac{1}{2} \quad \forall a_i \in \{0, 1\},
\end{equation} 
where it should be noted that, since the post-selection is performed in the generation rounds only, it does not introduce any loopholes in testing the NC inequality.
Furthermore, since we also assume that the spot-check and generation rounds produce identical correlations, the value $\beta$ observed in the spot-check rounds is associated to the generation rounds. 

\begin{theorem}
    For all quantum strategies $\tilde{S} = \left(\tilde{\rho}_A, \lbrace \tilde{\Pi}_i = \ket{\tilde{v}_i}\!\bra{\tilde{v}_i}_A\rbrace_{i = 1}^{N}\right)$ implemented by Alice that can be self-tested following Theorem~\ref{thm:theorem2}, after post-selection, the final cq-state shared between Alice and Eve, $\rho_{KE}$, satisfies
    \begin{equation}
        \norm{\rho_{KE} - 2^{-m}\mathds{1}_K \otimes \rho_E} \leq \mathcal{O}(m \sqrt{\epsilon}),
    \end{equation}
    where $\rho_E = \bigotimes_{j = 1}^{m} \ket{\xi_j}\!\bra{\xi_j}$ is the reduced state of Eve for the $m$ rounds where $\ket{\xi_j}$ is her state corresponding to the $j$th post-selected generation round.
    \label{thm:final_key}
\end{theorem} 

This recovers the security condition in Eq.~\eqref{eq:security} with $\epsilon_{sec} = \mathcal{O}(m\sqrt{\epsilon})$.

%%%%%%%%%%%%%%%%%%%%%%%%%%%%%%%%%%%%%%%%%%%%%%%%%%%%%%%%%%%%%%%%%%%%%%%%%%%%%%%%%%%%%%%%%%%%%%%%%%%%%%%%%%%%%%% 

\textit{Randomness expansion rate.---}The expansion rate $r$ quantifies the amount of randomness produced as compared to the amount consumed per round and is calculated by comparing their respective min-entropies. The final random bit string has length and smooth min-entropy $m$, with smoothing parameter $\delta = \mathcal{O}(m\sqrt{\epsilon})$. The value of $m$ is given by: 
\begin{equation}
    m = n - nq - \left(n - nq \right)\frac{1 - \omega_0}{1 + \cos\pi/N} = \frac{2n(1-q)\cos(\pi/N)}{1 +\cos(\pi/N)}.
\end{equation} 

This is computed by noting that the experiment consists of $n$ rounds, of which $\sum_{i} T_i$ are selected for spot-checking (with $\mathbb{E}\left(\sum_{i} T_i\right) = nq$) and that some of the generation rounds are discarded due to post-selection. 

The randomness consumed is calculated by noting that it is required for (a) generating the Bernoulli random variables $T_j$ for $j \in \{1, \ldots, n\}$ with bias $q$ using the interval algorithm~\cite{HH97}, (b) randomly choosing which measurements $\Pi_i, \Pi_{l'}$ to implement in each spot-check round (for more details, see `Random selection of measurements' in Ref.~\cite{HXA23}) (c) post-selecting the outcome $1$ in the generation rounds (which occurs with probability $1/\left[1 + \cos (\pi/N)\right]$) with probability $\omega_0$, again using the interval algorithm. In total, the expected amount of initial private randomness required, $l_{\mathrm{in}}$, for $n$ rounds is 
\begin{equation}
    l_{\mathrm{in}} \leq n \left[h(q) + q\log 2 N + (1 - q)p(0|1)  h\left(\omega_0\right) \right] + 6,
    \label{eq:consumed_rand}
\end{equation}
where $h(\cdot)$ is the binary entropy function and $p(0|1) = \frac{1}{1 + \cos{\frac{\pi}{N}}}$ is the probability to obtain the outcome $0$ when performing the measurement $\Pi_1$. The constant $6$ arises from using the interval algorithm twice, due to an expected consumption of at most $nh(q) + 3$ uniform bits when producing $n$ bits with $q$ bias.

The expected expansion rate can be bounded by $r \geq \frac{m - l_{\mathrm{in}}}{n}$.
We plot the expected randomness expansion rate as a function of the total number of experimental rounds in Fig.~\ref{fig:rand} with $\beta = \beta_{\text{QT}}$ for different values of $n$ and $N$. We note that the randomness expansion rate tends to 1 in the limit, as $q \to 0$ and $n, N \to \infty$.

%%%%%%%%%%%%%%%%%%%%%%%%%%%%%%%%%%%%%%%%%%%%%%%%%%%%%%%%%%%%%%%%%%%%%%%%%%%%%%%%%%%%%%%%
% Fig.
%%%%%%%%%%%%%%%%%%%%%%%%%%%%%%%%%%%%%%%%%%%%%%%%%%%%%%%%%%%%%%%%%%%%%%%%%%%%%%%%%%%%%%%%

\begin{figure}
    \centering
    \includegraphics[width = 0.47\textwidth]{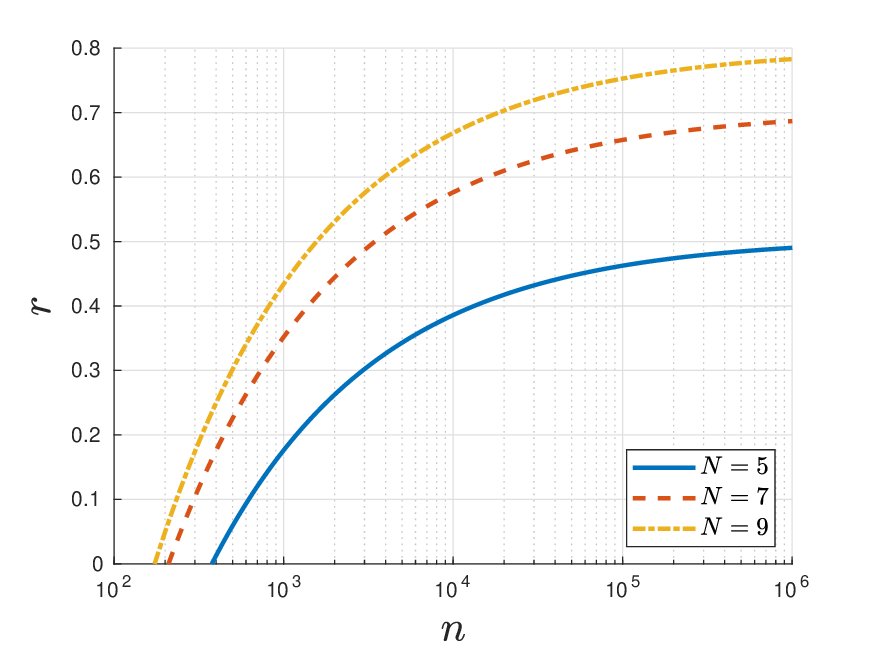}
    \caption{The expected net randomness expansion per round, $r$, as a function of total number of experimental rounds, $n$. We take the number of measurements as $N = 5$, $7$, and $9$, set the probability for spot-checking as $q = \frac{1}{\sqrt{n}}$ and set $\beta = \beta_{\text{QT}}$ (giving $\epsilon_{sec} = 0$).} 
    \label{fig:rand}
\end{figure}

%%%%%%%%%%%%%%%%%%%%%%%%%%%%%%%%%%%%%%%%%%%%%%%%%%%%%%%%%%%%%%%%%%%%%%%%%%%%%%%%%%%%%%%%

\textit{Experimental accessibility.---}As an example of RE using our scheme, we consider a recent experiment of the KCBS self-test~\cite{HXA23}. The maximum violation of the loophole-free KCBS inequality observed experimentally in Ref.~\cite{HXA23} (Table $2$) is $\beta = 2.236$, which yields $\epsilon \approx 10^{-5}$, from $10^4$ experimental rounds. If the experiment had followed our spot-checking QRE protocol and observed this value of $\beta$, the security parameter would be $\epsilon_{sec} = \mathcal{O}(10^{-2}m)$. This implies that, although the experiment could generate fresh random bits, the linear scaling with $m$ limits the final random bit string to approximately 100 bits before $\epsilon_{sec} = \mathcal{O}(1)$. As a result, RE cannot be achieved in this case, since the randomness consumed exceeds 100 bits. We discuss possible mitigations in the discussion section and provide similar analysis (for a non-robust version) for the even $N$-cycle scenarios in the Appendix. 

%%%%%%%%%%%%%%%%%%%%%%%%%%%%%%%%%%%%%%%%%%%%%%%%%%%%%%%%%%%%%%%%%%%%%%%%%%%%%%%%%%%%%%%%%%%%%%%%%%%%%%%%%%%%%%%
%\section{Discussion}
%\label{sec:conc}
%%%%%%%%%%%%%%%%%%%%%%%%%%%%%%%%%%%%%%%%%%%%%%%%%%%%%%%%%%%%%%%%%%%%%%%%%%%%%%%%%%%%%%%%%%%%%%%%%%%%%%%%%%%%%%%

\textit{Discussion---}We have presented a scheme for SDI-QRE via robust self-testing using NC inequalities that is secure against a quantum adversary. In particular, we
have shown that contextual correlations that result in a close-to-maximal violation of a NC inequality are sufficient to certify secure RE. Unlike many of the SDI schemes~\cite{UZZ13, UZZ20, PPF21}, our scheme does not require trusted measurements or preparation sources and is secure against a quantum adversary (and not only against a classical adversary). 
The advantage comes from using a different set of assumptions. Crucially, the assumptions required for the self-testing
can be (and have been) experimentally verified~\cite{LLS11, UZZ13, JRO16,ZKK17,MZL18, ZXX19, UZZ20, HXA23}.
Our approach extends self-testing techniques (used, so far, under DI assumptions) to SDI scenarios.

For achieving security, our QRE protocol requires providing certain additional experimental characterizations, including ensuring that the measurements are repeatable and satisfy no-disturbance, both of which are required in self-testing from NC inequalities. This increases the amount of required resources for the certification of security. However, this is expected as there is always a trade-off between the level of trust and security and the amount of resources required to achieve it. It may be argued that achieving correlations that offer close-to-maximal violation of NC inequalities may be a downside of our scheme, but as evidenced by the experiment in Ref.~\cite{HXA23}, it is already possible to achieve such correlations in the laboratory in a loophole-free manner under a similar set of assumptions that are required for our proposal. 

Similarly to other works on QRE based on Bell self-testing~\cite{WBC22}, our results only hold for close-to-maximum violation of the NC inequality. However, so far, no Bell experiment has achieved such correlations in a loophole-free manner~\cite{ZvR22, NDN22}.
In contrast, the KCBS self-test in Ref.~\cite{HXA23} and other recent experiments~\cite{PJC15, CLB15, SKC18} successfully achieved close-to-maximum violation of NC inequalities.
These experiments, when adapted to our protocol, and the observed correlations could be used to generate fresh random numbers. Although the violation of NC inequalities currently achieved in the lab does not allow randomness expansion (as the randomness consumed is more than generated), it indicates that experiments capable of performing randomness expansion with our protocol are close to being practically accessible. 

Several future research directions can be identified. Firstly, at present, our results only hold for an asymptotic number of experimental rounds, since all the relevant quantities are evaluated as an expected value. It would be interesting to extend this analysis to the case for a finite number of experimental rounds.
Secondly, the security parameter in our approach scales linearly with $m$, which is sub-optimal from a cryptographic perspective. This could be improved by relating (conditional) min-entropy to the violation of NC inequalities, similar to Bell inequalities, and then applying a randomness extractor \cite{M13, F24}. Alternatively, one could establish the self-testing statement directly over multiple rounds.
Finally, in a recent work~\cite{nath2024single}, the authors propose an alternate SDI-QRE protocol based on the Leggett-Garg inequality which is only secure against a classical adversary, while in Ref.~\cite{S23} the authors propose a variant of the SDI-QRE protocol based on Leggett-Garg inequality but prove its security against a quantum adversary. Thus, a promising future direction could involve applying our framework to other self-tests, particularly those based on self-testing of the Leggett-Garg inequality, as proposed in Ref.~\cite{shrotriya2022certifying}.

%%%%%%%%%%%%%%%%%%%%%%%%%%%%%%%%%%%%%%%%%%%%%%%%%%%%%%%%%%%%%%%%%%%%%%%%%%%%%%%%%%%%%%%%

\begin{acknowledgments}
{{\em Acknowledgements---}}This research is supported by A*STAR C230917003. J.~S. acknowledges support from the National Science and Technological Council (NSTC) of Taiwan through grant Nos.~NSTC 112-2628-M-006-007-MY4 and NSTC 112-2811-M-006-033-MY4. A.~C.\ is supported by the MCINN/AEI project ``New tools in quantum information and communication'' (Project No.\ PID2020-113738GB-I00) and the Digital Horizon Europe project \href{https://doi.org/10.3030/101070558}{``Foundations of quantum computational advantage'' (FoQaCiA)} (Grant agreement No.\ 101070558).
We thank F.~J.\ Curchod, V.~J.\ Wright, and L.\ Wooltorton for useful comments.
\end{acknowledgments}

%%%%%%%%%%%%%%%%%%%%%%%%%%%%%%%%%%%%%%%%%%%%%%%%%%%%%%%%%%%%%%%%%%%%%%%%%%%%%%%%%%%%%%%%

%\bibliography{references}


\newpage

\appendix

%%%%%%%%%%%%%%%%%%%%%%%%%%%%%%%%%%%%%%%%%%%%%%%%%%%%%%%%%%%%%%%%%%%%%%%%%%%%%%%%%%%%%%%%%%%%%%%%%%%%%%%%%%%%%%%
\onecolumngrid
%%%%%%%%%%%%%%%%%%%%%%%%%%%%%%%%%%%%%%%%%%%%%%%%%%%%%%%%%%%%%%%%%%%%%%%%%%%%%%%%%%%%%%%%%%%%%%%%%%%%%%%%%%%%%%%

\section{Non-contextuality inequalities}
\label{app:selftest}

%%%%%%%%%%%%%%%%%%%%%%%%%%%%%%%%%%%%%%%%%%%%%%%%%%%%%%%%%%%%%%%%%%%%%%%%%%%%%%%%%%%%%%%%%%%%%%%%%%%%%%%%%%%%%%%

In this section, we provide some details on the tight non-contextuality (NC) inequalities associated to the $N$-cycle scenarios, with $N \ge 4$ \cite{AQB13}, namely, scenarios whose graph of compatibility is a cycle on $N$ vertices. These scenarios are particularly important since a necessary condition for contextuality with ideal measurements is that the graph of compatibility of the scenario contains, induced, $N$-cycle scenarios, with $N \ge 4$ \cite{budroni2022kochen}. It is convenient to distinguish between the case in which $N$ is odd and the case in which $N$ is even.

For $N \ge 5$, with $N$ odd, the (only \cite{AQB13}) tight NC inequalities can be written as
\begin{equation}
    \mathcal{\beta} = \sum_{i = 1}^N p(1,0|i,i+1) \leq \beta_{\text{NCHV}},
    \label{oddineq}
\end{equation}
where the NC bound is
\begin{equation}
\beta_{\text{NCHV}} = \alpha(C_N) = \frac{N-1}{2},
\end{equation}
where $\alpha(C_N)$ denotes the independence number of $C_N$ \cite{CSW14}. Notice, in the case of Eq.~\eqref{oddineq}, $C_N$ is not only the graph of compatibility of the $N$ measurements in the scenario but also the graph of exclusivity of the $N$ events $(1,0|i,i+1)$ in \eqref{oddineq} (in a graph of exclusivity mutually exclusive events are represented by adjacent vertices). This second reason is why $\beta_{\text{NCHV}} = \alpha(C_N)$ \cite{CSW14}.
The maximum quantum value of $\beta$ is
\begin{equation}
\beta_{\text{QT}} = \vartheta(C_N) = \frac{N \cos \left(\pi/N\right)}{1 + \cos \left(\pi/N\right)},
    \label{eq:cyclic_lovasz}
\end{equation}
where $\vartheta(C_N)$ denotes the Lov\'asz number of $C_N$ \cite{CSW14}. $\beta_{\text{QT}}$ can be achieved with the following initial state
\begin{equation}
    \ket{v_0} = \left(1, 0, 0\right)^T,
\label{eq:cyclic_strategy}
\end{equation}
where $T$ indicates transposition, and the measurements 
\begin{equation}
    \Pi_i = \ket{v_i} \bra{v_i},\;\text{with }i \in \lbrace1, \ldots, N\rbrace,
\end{equation}
where
\begin{equation}
    \ket{v_i} = \left(\cos\theta, \sin\theta \sin\phi_i, \sin\theta \cos\phi_i\right)^T,
\end{equation}
with $\cos^2\theta = \frac{\cos \left(\pi/N\right)}{1 + \cos \left(\pi/N\right)}$ and $\phi_i = \frac{i \pi(N - 1)}{N}$. Notice that, for $i=1,\dots,N$ and taking $N+1=1$, $\bra{v_i}\ket{v_{i+1}}=0$, therefore $[\Pi_i,\Pi_{i+1}]=0$, i.e., $\Pi_i$ and $\Pi_{i+1}$ are jointly measurable.

For $N \ge 4$, with $N$ even, the (only \cite{AQB13}) tight NC inequalities are
\begin{equation}
    \mathcal{\gamma} = \left[ \sum_{i = 1}^{N-1} p(0,0|i,i+1) + p(1,1|i,i+1) \right] + p(0,1|N,1) + p(1,0|N,1) \leq \gamma_{\text{NCHV}},
\label{evenineq}
\end{equation}
where the NC bound is
\begin{equation}
\gamma_{\text{NCHV}} = \alpha(M_{2N}) = N-1,
\end{equation}
where $M_{2N}$ is the M\"obius ladder of order $2N$ \cite{AQB13}, which is the graph of exclusivity of the $2N$ events in \eqref{evenineq}. It should be noted that the exclusivity graph of these scenarios is not the same as the compatibility graph, which is cyclic.
The maximum quantum value is
\begin{equation}
\gamma_{\text{QT}} = \vartheta(M_{2N}) = \frac{N}{2} \left[1 + \cos \left(\pi/N\right)\right].
    \label{eq:cyclic_lovasz}
\end{equation}

$\gamma_{\text{QT}}$ can be achieved \cite{Pearle:1970PRD} with the two-qubit state
\begin{equation}
    \ket{v_0} = \frac{1}{\sqrt{2}}\left(1,0,0,1\right)^T,
\end{equation}
and the measurements $\Pi_i = \ket{v_i}\!\bra{v_i} = \frac{1}{2}\left[\mathds{1} + \bm{S}(i)\cdot\bm{\sigma}\right]$, which acts on first qubit if $i \in \lbrace 2, 4, \ldots, N \rbrace$ and on the second qubit if $i \in \lbrace 1, 3, \ldots,  N - 1 \rbrace$, $\bm{S}(i) = \left(\cos \left(i \pi/N\right), \sin \left(i \pi/N\right), 0\right)^T$ and $\bm{\sigma} = \left(\sigma_x, \sigma_y, \sigma_z\right)$.

%%%%%%%%%%%%%%%%%%%%%%%%%%%%%%%%%%%%%%%%%%%%%%%%%%%%%%%%%%%%%%%%%%%%%%%%%%%%%%%%%%

%Assumptions

%%%%%%%%%%%%%%%%%%%%%%%%%%%%%%%%%%%%%%%%%%%%%%%%%%%%%%%%%%%%%%%%%%%%%%%%%%%%%%%%%%

\section{Assumptions}
\label{app:assumptions}
In this section, we detail the assumptions used in deriving our results. Many of these follow from the assumptions necessary for self-testing from NC inequalities, while others are standard in several randomness expansion schemes.

\begin{enumerate}
    \item \label{it:one} 
    All measurements are repeatable, meaning they yield the same result when repeated on the same system.

    \item \label{it:two}  
    Compatible measurements do not disturb each other when performed on the same system. Specifically, if compatible measurements $\Pi_i$ and $\Pi_j$ are performed, then $p(a,b|i,j) = p(b,a|j,i)$ for all $a,b$. In other words, the order of these measurements does not affect the outcomes.

    \item \label{it:three} 
    The measurement and state preparation devices are memoryless. This means that the state prepared in each experimental round is independent of all previous rounds, and measurement outcomes do not depend on prior settings or outcomes.

    \item \label{it:five} 
    Alice has access to a short string of private random bits to input into the protocol. The expected length of this string is given by
    \begin{equation}
        l_{in} \leq n\left[h(q) + q \log 2N + (1 - q) p(0|1) h(\omega_0) \right] + 6.
        \label{eq_app:consumed_rand}
    \end{equation}
    
    \item \label{it:six} 
    Alice's lab is built such that no information about measurement outcomes leaks to the adversary.

    \item \label{it:seven} 
    Alice has a trusted and private classical computer to perform all necessary computations.

    \item \label{it:eight} 
    The correlations observed by Alice from her devices are identical in both the spot-check and generation rounds.

    \item \label{it:eight-b}All sources of errors and noise impact all the experimental rounds equally.

    \item \label{it:nine} 
    Quantum theory is correct and complete.
\end{enumerate}

Assumptions~\ref{it:one}-\ref{it:three} are necessary for self-testing of NC inequalities, and assumption \ref{it:three} is also crucial for the derivation of our main result. Assumption~\ref{it:five} is necessary for the implementation of the protocol. Assumptions~\ref{it:six}-\ref{it:seven} ensure that the privacy of the RE protocol is not trivially compromised, by preventing the devices from directly communicating outcomes to the adversary. Assumption~\ref{it:eight} allows us to associate generation rounds with a specific experimental value of a NC inequality, assumption \ref{it:eight-b} allows us to relate the self-testing statement to each individual round and assumption~\ref{it:nine} ensures the theoretical framework’s correctness, whilst constraining the adversary's behavior.

The assumptions~\ref{it:one}-\ref{it:two} (required for the self-test) can be experimentally verified. Assumption~\ref{it:one} can be verified by evaluating $R = p(0,0|i,i) + p(1,1|i,i)$, where $p(a,a|i,i)$ denotes the probability of obtaining outcome $a$ twice when measurement setting $i$ is performed twice. If $ R = 1 $, the assumption holds exactly; if $R$ is close to $1$, the assumption holds approximately, as seen in experiments such as those in Refs.~\cite{Kirchmair:2009NAT,MZL18,WZL22,HXA23}. 
then the assumption holds exactly. In a similar manner, Assumption~\ref{it:two} can be tested by sequentially by implementing measurements which are compatible, to obtain $p(a, b|i,j)$, and then reversing the order of the measurements, to obtain $p(b,a|j,i)$. If the quantity $p(a, b|i,j) - p(b,a|j,i) = 0$ then the assumption holds exactly; if it is close to $0$, then the assumption holds approximately, as seen in Ref.~\cite{HXA23}.

Assumption~\ref{it:three} ensures the correctness of the self-test and guarantees that experimental rounds are independent. 
We note that there exist some methods that can help bypass this assumption~\cite{DF19, KZF18}, but we do not make use of them here. 
Specifically, it has been shown that the amount of private entropy produced when devices possess memory is comparable to that in independent and identically distributed (i.i.d.) scenarios, given sufficiently many rounds. In the context of randomness expansion from contextual correlations, new techniques may help address issues related to limited memory or small signaling between rounds~\cite{VEB23, FBM24}.

Assumption~\ref{it:eight} is undesirable for finite round experiments, but can be relaxed using an appropriate concentration inequality, as demonstrated in Refs.~\cite{MRCWB14, FM24}.

It is important to note that we do \textit{not} assume the state prepared by Alice is pure. However, it is known that only (close-to) pure states can exhibit the violations of NC inequalities necessary for our results. This restriction does not apply to the measurements, as we already assume they are repeatable and that their statistics satisfy the required relationships.

%%%%%%%%%%%%%%%%%%%%%%%%%%%%%%%%%%%%%%%%%%%%%%%%%%%%%%%%%%%%%%%%%%%%%%%%%%%%%%%%%%%%%%%%%%%%%%%%%%%%%%%%%%%%%%%
\section{Proofs of main theorems}
\label{app:proofs}
\setcounter{theorem}{0}
\setcounter{definition}{0}
In this section, we prove our main theorems. First, we show that for maximum violation of an NC inequality, i.e., $\beta = \beta_{\text{QT}}$, Alice's single-round post-measurement state is uncorrelated with Eve. Next, we use similar techniques to prove our results in the case of near-maximal violations, i.e., when $\beta \geq \beta_{\text{QT}} - \epsilon$, where $\epsilon \ll 1$. Finally, after establishing these relations at the level of single-round post-measurement states, we show that the final post-measurement state of Alice's post-selected final random bit string is $\mathcal{O}(m\sqrt{\epsilon})$ close to uniform, where $m$ is the length of the final random bit string generated.

In what follows, a quantum strategy $\mathcal{S}$, consisting of a set of quantum states and measurements, is denoted by a tuple, e.g., $\mathcal{S} = (\rho, \{ \Pi_i \}_{i=1}^{N})$, where, $\rho$ represents the quantum state used, and $\Pi_i$ denotes Alice's projective measurements having outcomes $a_i\in \lbrace 0,1 \rbrace$ and the post-measurement states by $\ket{a_i}\!\bra{a_i}_A$ where $a_i$ denotes the outcome and the subscript $i$ denotes the measurement setting and $\ket{0_i}_A = \ket{v_i}_A$.  In our case, the set of $\Pi_i$ are those that appear in the relevant NC inequality. Strategies and measurement outcomes with a tilde denote those that Alice may be implementing, which may or may not reach the maximum quantum value. In contrast, strategies without tildes are optimal and achieve $\beta = \beta_{\text{QT}}$.

First, we re-state the definition of self-testing of NC correlations from the main text, which will be used throughout the proofs of the theorems.

\begin{definition}[Robust self-testing from NC correlations; rephrased from Ref.~\cite{BRV19}, ]
\label{def:self_test}
    If $\mathcal{G} = (\mathcal{V}, \mathcal{E})$ is a cycle with $N \ge 5$ vertices and $N$ odd, and a quantum strategy $\mathcal{S}=\left(\ket{v_0}\!\bra{v_0}, \lbrace \Pi_i = \ket{v_i}\!\bra{v_i}\rbrace_{i = 1}^{N}\right)$, where $\ket{v_0}\!\bra{v_0}$ is the initial state, $\lbrace v_i\rbrace_{i = 1}^{N} \in \mathcal{V}$ and  $\text{tr}\left(\Pi_i \Pi_j\right) = 0$ for all $(v_i, v_j)\in \mathcal{E}$, achieves
$\beta = \beta_{\text{QT}}$, then the NC inequality Eq.~\eqref{eq:nc_inequality} serves as a robust local self-test for this strategy in the sense that for any other quantum strategy $\tilde{\mathcal{S}}= \left(\ket{\tilde{v}_0}\!\bra{\tilde{v}_0}, \lbrace \tilde{\Pi}_i = \ket{\tilde{v}_i}\!\bra{\tilde{v}_i}\rbrace_{i = 1}^{N}\right)$ that achieves $\beta = \beta_{\text{QT}} - \epsilon$, there exists an isometry $V_A$ such that $\norm{V_A\ket{\tilde{v}_i}\!\bra{\tilde{v}_i}V^\dagger_A - \ket{v_i}\!\bra{v_i}} \leq \mathcal{O}(\sqrt{\epsilon})$ for $i = 0, 1, \ldots, N$ and $V^\dagger_A V_A = \mathds{1}$.
\end{definition} 

Here, $\norm{A} = \text{tr}\left(\sqrt{A^\dagger A}\right)$ denotes the trace norm of a matrix $A$.

Next, we prove the theorems appearing in the main text. We begin with the case when $\beta = \beta_{\text{QT}}$ and derive a specific extraction map (isometric transformation that relates the implemented strategy to the ideal one) that will be used throughout all the proofs.

%%%%%%%%%%%%%%%%%%%%%%%%%%%%%%%%%%%%%%%%%%%%%%%%%%%%%%%%%%%%%%%%%%%%%%%%%%%%%%%%%%

%Proof of theorem 1

%%%%%%%%%%%%%%%%%%%%%%%%%%%%%%%%%%%%%%%%%%%%%%%%%%%%%%%%%%%%%%%%%%%%%%%%%%%%%%%%%%

\begin{theorem}
\label{thm:theorem1}
    For any quantum strategy $\tilde{\mathcal{S}} = \left(\tilde{\rho}_A, \lbrace \tilde{\Pi}_i = \ket{\tilde{v}_i}\!\bra{\tilde{v}_i}_A\rbrace_{i = 1}^{N}\right)$ implemented by Alice that achieves $\beta = \beta_{\text{QT}}$, the adversary Eve is uncorrelated with her post-measurement state, which can be written as
    \begin{equation}
    \rho_{AE} = \sum_{a_i = 0}^{1} p(a_i|i) \ket{a_i}\!\bra{a_i}_A \otimes \ket{\xi}\!\bra{\xi}_E,
\end{equation}
    where $\rho_{AE}$ denotes the joint state of Alice and Eve and $\ket{\xi}_E$ denotes the system held by Eve.
\end{theorem}

\begin{proof}
    The basic qualitative idea for the proof is to show that the state and measurements of Alice are uncorrelated from Eve's.
    
    We begin by considering the purification $\ket{\psi}_{AE}$ of Alice's reduced state $\tilde{\rho}_A$, which can be written as
    \begin{equation}
        \ket{\psi}_{AE} = \sum_l \sqrt{c_l} \ket{\phi_l}_A \otimes \ket{l}_E,
    \end{equation}
    where $c_l \geq 0$ for all $l$, $\lbrace \ket{l}_E \rbrace$ form an orthonormal basis for Eve's Hilbert space and the state $\tilde{\rho}_A = \text{tr}_E\left(\ket{\psi}\!\bra{\psi}_{AE}\right) = \sum_l c_l \ket{\phi_l}\!\bra{\phi_l}$.
    Following Remark~\ref{def:self_test}, if Alice observes maximum violation of the NC inequality then the following relations hold: $\ket{\phi_l}_A = V_A\ket{v_0}_A$ and $\ket{\tilde{v}_i}\!\bra{\tilde{v}_i}_A = V_A \ket{v_i} \bra{v_i}_AV_A^\dagger$ for all $i\in \lbrace 1,\ldots, N \rbrace$, where $V_A$ is an isometry acting on Alice's system only and satisfies $V^\dagger_A V_A = \mathds{1}$. Next, we define the isometric transformation
    \begin{equation}
        U = V_A^\dagger \otimes \mathds{1}_E,
        \label{eq:iso_app}
    \end{equation}
    and call this our extraction map. Using this extraction map $U$, we show that
    \begin{equation}
        \begin{aligned}
            &U\left(\ket{\tilde{v}_i}\!\bra{\tilde{v}_i}_A \otimes \mathds{1}_E\right)U^\dagger U \ket{\psi}_{AE}\\
            &\quad = \sum_{l}\sqrt{c_l} \left(V^\dagger_A \ket{\tilde{v}_i}\!\bra{\tilde{v}_i}_A V_A \otimes \mathds{1}_E\right) \left( V^\dagger_A \ket{\phi_l}_A \otimes \ket{l}_E\right)\\
            &\quad = \left(\ket{v_i}\!\bra{v_i}_A \otimes \mathds{1}_E\right) \left( \sum_l \sqrt{c_l} \ket{v_0}_A \otimes \ket{l}_E\right)\\
            &\quad = \left( \Pi_i \ket{v_0}_A\right) \otimes \ket{\xi}_E,
        \end{aligned}
        \label{eq:self_test_extraction1}
    \end{equation}
    for all $i$, where in the final equality we have substituted $\ket{\xi}_E = \sum_l \sqrt{c_l} \ket{l}_E$.

    In a similar fashion to above, we can show that
    \begin{equation}
        U\left[\left(\mathds{1} - \tilde{\Pi}_i\right)\otimes \mathds{1}_E \ket{\psi}_{AE}\right] = \left[\left(\mathds{1} - \Pi_i\right) \ket{v_0}_A\right] \otimes \ket{\xi}_E.
        \label{eq:self_test_extraction2}
    \end{equation}

    We can collectively write Eqs.~\eqref{eq:self_test_extraction1} and \eqref{eq:self_test_extraction2} for all permissible values of $a$ and $i$ as
    \begin{equation}
        U\left(\ket{\tilde{a}_i}\!\bra{\tilde{a}_i}_A \otimes \mathds{1}_E\right) \ket{\psi}_{AE} = \left(\ket{a_i}\!\bra{a_i}_A \otimes \mathds{1}_E \right)\ket{v_0}_A \otimes \ket{\xi}_E. 
    \end{equation}

Next, we can evaluate the post-measurement cq-state of Alice and Eve, which is given by
\begin{equation}
    \rho_{AE} = \sum_{\tilde{a_i} = 0}^{1} \ket{\tilde{a}_i}\!\bra{\tilde{a}_i}_A \otimes \rho_E^{\tilde{a}_i} \quad \forall i,
    \label{eq:cq_state1}
\end{equation}
where $\rho_E^{\tilde{a}_i} = \text{tr}_A \left[\left(\ket{\tilde{a}_i}\!\bra{\tilde{a}_i}_A \otimes \mathds{1}_E\right) \ket{\psi}\!\bra{\psi}_{AE}\right]$.
    
For the moment, let us focus on the term $\text{tr}_A\left[\left(\tilde{\Pi}_i \otimes \mathds{1}_E\right) \ket{\psi}\!\bra{\psi}_{AE}\right]$ and using the fact that the isometry $U^\dagger U = \mathds{1}$ (since $V_A$ is unitary), we have
\begin{equation}
    \begin{aligned}
        &\text{tr}_A\left[\left(\tilde{\Pi}_i \otimes \mathds{1}_E\right) \ket{\psi}\!\bra{\psi}_{AE}\right] \\
        &\quad =  \text{tr}_A \left[U^\dagger U \left(\tilde{\Pi}_i \otimes \mathds{1}_E\right) \ket{\psi}\!\bra{\psi}_{AE} U^\dagger U\right]\\
        &\quad = \text{tr}_A \left[U^\dagger \left(\Pi_i \ket{v_0}\!\bra{v_0}_A \otimes \ket{\xi}\!\bra{\xi}_E \right)U\right]\\
        &\quad = \text{tr} \left(\Pi_i \ket{v_0}\!\bra{v_0}_A\right) \ket{\xi}\!\bra{\xi}_E\\
        &\quad = p(1|i) \ket{\xi}\!\bra{\xi}_E,
    \end{aligned} 
    \label{eq:cq_final_state1}
\end{equation}
where ${\rm tr}\left(\Pi_i \ket{v_0}\!\bra{v_0}_A\right) = p(0|i)$. It should be noted that here Eve's state is independent of Alice's measurement outcome.
Performing a similar analysis for the term $\text{tr}_A\left[\left(\mathds{1} - \tilde{\Pi}_i \otimes \mathds{1}_E\right) \ket{\psi}\!\bra{\psi}_{AE}\right]$ we obtain
\begin{equation}
    \text{tr}_A\left[\left(\mathds{1} - \tilde{\Pi}_i \otimes \mathds{1}_E\right) \ket{\psi}\!\bra{\psi}_{AE}\right] = p(0|i) \ket{\xi}\!\bra{\xi}_E.
    \label{eq:cq_final_state2}
\end{equation}

Substituting Eqs.~\eqref{eq:cq_final_state1} and \eqref{eq:cq_final_state2} in Eq.~\eqref{eq:cq_state1} results in the post-measurement state being 
\begin{equation}
    \rho_{AE} = \sum_{a = 0}^{1} p(a_i|i) \ket{a_i}\!\bra{a_i}_A \otimes \ket{\xi}\!\bra{\xi}_E,
\end{equation}
where it can be seen that Eve is now uncorrelated with the results of Alice's measurements. This concludes the proof.
\end{proof}

Next, we move on to prove our results in the case when Alice observes close-to-maximum violation of the odd $N$-cycle NC inequality. Note that we will still use the same isometric transformation as defined in Eq.~\eqref{eq:iso_app}.

%%%%%%%%%%%%%%%%%%%%%%%%%%%%%%%%%%%%%%%%%%%%%%%%%%%%%%%%%%%%%%%%%%%%%%%%%%%%%%%%%%

%Proof of theorem 2

%%%%%%%%%%%%%%%%%%%%%%%%%%%%%%%%%%%%%%%%%%%%%%%%%%%%%%%%%%%%%%%%%%%%%%%%%%%%%%%%%%

\begin{theorem}
    \label{thm:theorem2}
    For any quantum strategy $\tilde{\mathcal{S}} = \left(\tilde{\rho}_A, \lbrace \tilde{\Pi}_i = \ket{\tilde{v}_i}\!\bra{\tilde{v}_i}_A\rbrace_{i = 1}^{N}\right)$ implemented by Alice that achieves $\beta = \beta_{\text{QT}} - \epsilon$ $(\epsilon \ll 1)$, the post-measurement state of Alice and Eve, $\rho_{AE}$, is $\mathcal{O}( \sqrt{\epsilon})$-close (in trace distance) to a state which is uncorrelated to Eve, i.e.,
    \begin{equation}
        \norm{\rho_{AE} - \sum_{a = 0}^{1} p(a_i|i) \ket{a_i}\!\bra{a_i}_A \otimes \ket{\xi}\!\bra{\xi}_E} \leq \mathcal{O}(\sqrt{\epsilon}),
    \label{eq:aforementioned}
    \end{equation}
    where $\ket{\xi}_E$ denotes the system held by Eve.
    \end{theorem}

\begin{proof}
    The main idea of the proof is to show that Alice's measurements and state are close to being uncorrelated from Eve's. The proof can be broken down into four major steps.
    We first briefly describe these steps qualitatively and then move on to the formal derivation:
    \begin{enumerate}
        \item In the first step we bound the trace distance between the implemented measurement strategy of Alice, denoted by tildes and the ideal measurement strategy which achieves $\beta = \beta_{\text{QT}}$ for the same joint state of Alice and Eve denoted by $\ket{\psi}_{AE}$. For all permissible values of $a_i$, $\tilde{a}_i$ and $i$ we obtain
    \begin{equation}
            \norm{\begin{aligned}
            &U\left(\ket{\tilde{a}_i}\!\bra{\tilde{a}_i}_A \otimes \mathds{1}_E\right)U^\dagger \ket{\psi}\!\bra{\psi}_{AE}\\
            &~~\quad - \left(\ket{a_i}\!\bra{a_i}_A \otimes \mathds{1}_E\right)\ket{\psi}\!\bra{\psi}_{AE} 
            \end{aligned}}\leq \mathcal{O}(\sqrt{\epsilon}),
    \label{eq:self_test_extract_meas}
    \end{equation}
    
    \item In the second step we bound the distance between the joint state of Alice and Eve (which may be entangled), and the ideal state of Alice (which achieves $\beta = \beta_{\text{QT}}$ for the set of ideal measurement settings as in Remark~\ref{def:self_test}), which is uncorrelated from Eve's state. In this step we obtain
    \begin{equation}
        \norm{\left(\ket{a_i}\!\bra{a_i}_A \otimes \mathds{1}_E\right)\ket{\psi}\!\bra{\psi}_{AE} - \left(\ket{a_i}\!\bra{a_i}_A\otimes \mathds{1}_E\right)\left(\ket{v_0}\!\bra{v_0}_A\otimes \ket{\xi}\!\bra{\xi}_E\right)} \leq \mathcal{O}(\sqrt{\epsilon}).
        \label{eq:self_test_state_required}
    \end{equation}

    \item Then in the third step we apply the triangle property of the trace distance and using Eqs.~\eqref{eq:self_test_extract_meas} and \eqref{eq:self_test_state_required} to compile the aforementioned two steps together and obtain
    \begin{equation}
        \norm{
            U\left(\ket{\tilde{a}_i}\!\bra{\tilde{a}_i}_A \otimes \mathds{1}_E\right)U^\dagger \ket{\psi}\!\bra{\psi}_{AE}
                - \left(\ket{a_i}\!\bra{a_i}_A \otimes \mathds{1}_E\right) \left(\ket{v_0}\!\bra{v_0}_A \otimes \ket{\xi}\!\bra{\xi}_E\right) 
            }\leq \mathcal{O}(\sqrt{\epsilon}).
    \end{equation}

    \item Finally, we use the aforementioned results and the approach from our proof of Theorem~\ref{thm:theorem1} to obtain Eq.~\eqref{eq:aforementioned}.
    \end{enumerate}
    
    We start with the first step to bound the distance between the implemented and ideal measurement strategies of Alice, while keeping the joint state of Alice and Eve to be the same. In this step we note that
    \begin{equation}
    \begin{aligned}
            &\norm{\begin{aligned}
            &U\left(\ket{\tilde{a}_i}\!\bra{\tilde{a}_i}_A \otimes \mathds{1}_E\right)U^\dagger \ket{\psi}\!\bra{\psi}_{AE}\\
            &~~\quad - \left(\ket{a_i}\!\bra{a_i}_A \otimes \mathds{1}_E\right)\ket{\psi}\!\bra{\psi}_{AE} 
            \end{aligned}}\\
            & \leq \norm{U\left(\ket{\tilde{a}_i}\!\bra{\tilde{a}_i}_A \otimes \mathds{1}_E\right)U^\dagger - \ket{a_i}\!\bra{a_i}_A\otimes \mathds{1}_E}\\
            & = \norm{V^\dagger_A\ket{\tilde{a}_i}\!\bra{\tilde{a}_i}_A V_A - \ket{a_i}\!\bra{a_i}_A} \\
            & \leq \mathcal{O}(\sqrt{\epsilon}),
    \end{aligned} 
    \label{eq:self_test_state_app}
    \end{equation}
where in the first inequality we have used the fact that $\norm{\ket{\psi}\!\bra{\psi}_{AE}} \leq 1$, in the second inequality we use the definition $U  = V^\dagger_A \otimes \mathds{1}_E$ and in the last line we used the definition of the robust self-test of NC inequalities.
    
    For the second step in the proof we note that the reduced state of Alice (upto a local unitary transformation) can be re-written as $\tilde{\rho}_A = \sum_ip_i \tilde{\rho}_i$, where $\lbrace \tilde{\rho}_i\rbrace$ form an eigenbasis for $\tilde{\rho}_A$.
    Now, let $\ket{\psi}\!\bra{\psi}_{AE} = U \ket{\phi}\!\bra{\phi}_{AE}U^\dagger$, where (upto local unitary transformations) $\ket{\phi}\!\bra{\phi}_{AE} = \sum_i p_i \tilde{\rho}_i\otimes \rho_E^i$ be the joint state of Alice and Eve. Then, from the term on the left hand side in  Eq.~\eqref{eq:self_test_state_required} we have
    \begin{equation}
    \begin{aligned}
        &\norm{\left(\ket{a_i}\!\bra{a_i} \otimes \mathds{1}_E\right)\ket{\psi}\!\bra{\psi}_{AE} - \left(\ket{a_i}\!\bra{a_i}\otimes \mathds{1}_E\right)\left(\ket{v_0}\!\bra{v_0}_A\otimes \ket{\xi}\!\bra{\xi}_E\right)} \\
        &=\left[\norm{\left(\ket{a_i}\!\bra{a_i} \otimes \mathds{1}_E\right)}\right] \left[\norm{\ket{\psi}\!\bra{\psi}_{AE} -\left(\ket{v_0}\!\bra{v_0}_A\otimes \ket{\xi}\!\bra{\xi}_E\right)}\right] \\
        &\leq \norm{\ket{\psi}\!\bra{\psi}_{AE} - \ket{v_0}\!\bra{v_0}_A\otimes \ket{\xi}\!\bra{\xi}}\\
        &=\norm{U\ket{\phi}\!\bra{\phi}_{AE}U^\dagger - \ket{v_0}\!\bra{v_0}_A\otimes \ket{\xi}\!\bra{\xi}}\\
        &= \norm{U\left(\sum_i p_i \tilde{\rho}_i \otimes \rho_E^i\right)U^\dagger - \ket{v_0}\!\bra{v_0}_A\otimes \ket{\xi}\!\bra{\xi}}\\
        &\leq \sum_i p_i\norm{V_A \tilde{\rho}_i V_A \otimes \rho_E^i - \ket{v_0}\!\bra{v_0}_A \otimes \ket{\xi}\!\bra{\xi}_E}\\
        &\leq \sum_i p_i\left(\norm{V_A \tilde{\rho}_i V_A \otimes \rho_E^i - \ket{v_0}\!\bra{v_0}_A \otimes \rho_E^i} + \norm{\ket{v_0}\!\bra{v_0}_A \otimes \rho_E^i - \ket{v_0}\!\bra{v_0}_A \otimes \ket{\xi}\!\bra{\xi}_E}\right)\\
        &\leq \sum_i p_i \norm{V_A \tilde{\rho}_i V_A \otimes \rho_E^i - \ket{v_0}\!\bra{v_0}_A \otimes \rho_E^i}\\
        &\leq \sum_i p_i \norm{V_A \tilde{\rho}_i V_A - \ket{v_0}\!\bra{v_0}_A}\\
        &\leq \sum_i p_i \mathcal{O}(\sqrt{\epsilon})\\
        &= \mathcal{O}(\sqrt{\epsilon}),
    \end{aligned}
    \label{eq:self_test_state}
    \end{equation}
    where in the first inequality we have used $\norm{\ket{a_i}\!\bra{a_i}_A\otimes \mathds{1}_E} \leq 1$, in the second inequality we have used the fact that trace distance is convex in each of its arguments. In the third inequality we added and subtracted the term $\ket{v_0}\!\bra{v_0} \otimes \rho_E^i$ inside the norm and then used the triangle property of the trace distance which states that for any three operators $\rho$, $\tau$ and $\sigma$, we have $\norm{\rho- \sigma} \leq \norm{\rho-\tau} + \norm{\tau- \sigma}$. In the fourth inequality we use the fact that trace distance is always positive.

    Next, in the third step we add Eqs.~\eqref{eq:self_test_state_app} and \eqref{eq:self_test_state} together and apply the triangle property of trace distance to obtain
    \begin{equation}
        \begin{aligned}
            &\norm{\begin{aligned}
                &U\left(\ket{\tilde{a}_i}\!\bra{\tilde{a}_i}_A \otimes \mathds{1}_E\right)U^\dagger
                \ket{\psi}\!\bra{\psi}_{AE}\\
                &~~\quad - \left(\ket{a_i}\!\bra{a_i}_A \otimes \mathds{1}_E\right)\ket{\psi}\!\bra{\psi}_{AE} 
            \end{aligned}}
            + \norm{\begin{aligned}
                &\left(\ket{a_i}\!\bra{a_i}_A \otimes \mathds{1}_E\right)\ket{\psi}\!\bra{\psi}_{AE}\\
                &~~\quad - \left(\ket{a_i}\!\bra{a_i}_A \otimes \mathds{1}_E\right)\left(\ket{v_0}\!\bra{v_0}_A \otimes \ket{\xi}\!\bra{\xi}_E\right) 
            \end{aligned}}\\
            &\geq \norm{\begin{aligned}
                &U\left(\ket{\tilde{a}_i}\!\bra{\tilde{a}_i}_A \otimes \mathds{1}_E\right)U^\dagger \ket{\psi}\!\bra{\psi}_{AE}\\
                &~~\quad - \left(\ket{a_i}\!\bra{a_i}_A \otimes \mathds{1}_E\right) \left(\ket{v_0}\!\bra{v_0}_A \otimes \ket{\xi}\!\bra{\xi}_E\right)
            \end{aligned}}
        \end{aligned}.
    \end{equation}

    However, the sum on the left hand side of the equation is upper bounded by $\mathcal{O}(\sqrt{\epsilon})$. Substituting this in the above equation yields the result
    \begin{equation}
        \norm{\begin{aligned}
                &U\left(\ket{\tilde{a}_i}\!\bra{\tilde{a}_i}_A \otimes \mathds{1}_E\right)U^\dagger \ket{\psi}\!\bra{\psi}_{AE}\\
                &~~\quad - \left(\ket{a_i}\!\bra{a_i}_A \otimes \mathds{1}_E\right) \left(\ket{v_0}\!\bra{v_0}_A \otimes \ket{\xi}\!\bra{\xi}_E\right)
            \end{aligned}} \leq \mathcal{O}(\sqrt{\epsilon})
            \label{eq:final2}
    \end{equation}

Next, for the final step, we note a result from Ref.~\cite{R12} which states that for an operator $X \in \mathcal{H}_1 \otimes \mathcal{H}_2$, where $\mathcal{H}_i$ denotes the $i$-th Hilbert space, the norm satisfies $\norm{\text{tr}_1\left(X\right)} \leq \norm{X}$, where $\text{tr}_1 (\cdot)$ is the partial trace over the system $1$. Following Eq.~\eqref{eq:final2} we have (for all permissible values of $a_i$ and $i$)
\begin{equation}
    \norm{\begin{aligned}
    &\text{tr}_A\left[U\left(\ket{\tilde{a}_i}\!\bra{\tilde{a}_i}_A\otimes \mathds{1}_E\right)U^\dagger\ket{\psi}\!\bra{\psi}_{AE} \right]\\
    &~~ \quad -\text{tr}_A \left[\left(\ket{a_i}\!\bra{a_i}_A\otimes \mathds{1}_E\right) \left(\ket{v_0}\!\bra{v_0}_A \otimes \ket{\xi}\!\bra{\xi}_E\right)\right] 
    \end{aligned}}\leq  \mathcal{O}(\sqrt{\epsilon}).
\end{equation}

However, the term on the left hand side can be re-written as
\begin{equation}
\begin{aligned}
    &\norm{\begin{aligned}
    &\text{tr}_A\left[U\left(\ket{\tilde{a}_i}\!\bra{\tilde{a}_i}_A\otimes \mathds{1}_E\right)U^\dagger\ket{\psi}\!\bra{\psi}_{AE} \right]\\
    &~~ \quad -\text{tr}_A \left[\left(\ket{a_i}\!\bra{a_i}_A\otimes \mathds{1}_E\right) \left(\ket{v_0}\!\bra{v_0}_A \otimes \ket{\xi}\!\bra{\xi}_E\right)\right] 
    \end{aligned}}\\
    & = \norm{\text{tr}_A\left[U\left(\ket{\tilde{a}_i}\!\bra{\tilde{a}_i}_A\otimes \mathds{1}_E\right)U^\dagger\ket{\psi}\!\bra{\psi}_{AE} \right]
    -p(a_i|i)\ket{\xi}\!\bra{\xi}_E 
    }\\
    & = \norm{\begin{aligned}
        &\ket{a_i}\!\bra{a_i}_A \otimes \text{tr}_A\left[U\left(\ket{\tilde{a}_i}\!\bra{\tilde{a}_i}_A\otimes \mathds{1}_E\right)U^\dagger\ket{\psi}\!\bra{\psi}_{AE} \right]\\
    &~~\quad - p(a_i|i)\ket{a_i}\!\bra{a_i}_A\otimes \ket{\xi}\!\bra{\xi}_E 
    \end{aligned}},
\end{aligned}
\end{equation}
where in the last equality, we have taken a tensor product of the entire expression with $\ket{a_i}\!\bra{a_i}_A$. Therefore, after the re-writing and for all permissible values of $a_i$ and $i$, we have
\begin{equation}
    \norm{\begin{aligned}
        &\ket{a_i}\!\bra{a_i}_A \otimes \text{tr}_A\left[U\left(\ket{\tilde{a}_i}\!\bra{\tilde{a}_i}_A\otimes \mathds{1}_E\right)U^\dagger\ket{\psi}\!\bra{\psi}_{AE} \right]\\
    &~~\quad - p(a_i|i)\ket{a_i}\!\bra{a_i}_A\otimes \ket{\xi}\!\bra{\xi}_E 
    \end{aligned}} \leq  \mathcal{O}(\sqrt{\epsilon}).
\end{equation} 

Next, by using the property of sub-additivity of the trace distance and taking a summation over all possible values of $a_i$ we obtain 
\begin{equation}
\begin{aligned}
    &\norm{\rho_{AE}
     - \sum_{a_i = 0}^{1} p(a_i|i)\ket{a_i}\!\bra{a_i}_A\otimes \ket{\xi}\!\bra{\xi}_E 
   }\\
   & \leq \sum_{a_i = 0}^{1}\norm{\ket{a_i}\!\bra{a_i}_A \otimes \text{tr}_A\left[U\left(\ket{\tilde{a}_i}\!\bra{\tilde{a}_i}_A\otimes \mathds{1}_E\right)U^\dagger\ket{\psi}\!\bra{\psi}_{AE} \right]
     - p(a_i|i)\ket{a_i}\!\bra{a_i}_A\otimes \ket{\xi}\!\bra{\xi}_E}\\
     &\leq \mathcal{O}( \sqrt{\epsilon}),
   \end{aligned}
   \label{eq:final_spot_check}
\end{equation}
where we have substituted $\rho_{AE} =\sum_{a_i = 0}^{1} \ket{a_i}\!\bra{a_i}_A \otimes \text{tr}_A\left[U\left(\ket{\tilde{a}_i}\!\bra{\tilde{a}_i}_A\otimes \mathds{1}_E\right)U^\dagger\ket{\psi}\!\bra{\psi}_{AE} \right]$ and taken the summation out of the norm (following the property of sub-additivity of trace distance). This completes the proof. 
\end{proof}

%%%%%%%%%%%%%%%%%%%%%%%%%%%%%%%%%%%%%%%%%%%%%%%%%%%%%%%%%%%%%%%%%%%%%%%%%%%%%%%%%%

%Proof of theorem 3

%%%%%%%%%%%%%%%%%%%%%%%%%%%%%%%%%%%%%%%%%%%%%%%%%%%%%%%%%%%%%%%%%%%%%%%%%%%%%%%%%%

\begin{theorem}
    For all quantum strategies $\tilde{S} = \left(\tilde{\rho}_A, \lbrace \tilde{\Pi}_i = \ket{\tilde{v}_i}\!\bra{\tilde{v}_i}_A\rbrace_{i = 1}^{N}\right)$ implemented by Alice that can be self-tested following Theorem~\ref{thm:theorem2}, after post-selection, the final cq-state shared between Alice and Eve, $\rho_{KE}$, satisfies
    \begin{equation}
        \norm{\rho_{KE} - 2^{-m}\mathds{1}_K \otimes \rho_E} \leq \mathcal{O}(m \sqrt{\epsilon}),
    \end{equation}
    where $\rho_E = \bigotimes_{j = 1}^{m} \ket{\xi_j}\!\bra{\xi_j}$ is the reduced state of Eve for the $m$ rounds where $\ket{\xi_j}$ is her state corresponding to the $j$th post-selected generation round.
    \label{thm:final_key}
\end{theorem} 

\begin{proof}
    The key ideas in the following proof are: 1) the assumption of memoryless measurement and state preparation devices ensures Alice's post-measurement state is independent of previous rounds imposing a product structure and 2) recursively applying the triangle inequality to the trace norm to extract single-round expressions, which can be bounded via the single round bound presented in Theorem~\ref{thm:theorem2}.

    Let $\rho_{K_jE}$ and $k_j$ denote the post-measurement $j$th state of Alice and the ideal $j$th key register uncorrelated with Eve, respectively. The proof follows by recursively applying the triangle inequality $m$ times. We start  
    \begin{equation}
        \begin{aligned}
            &\norm{\rho_{KE} - 2^{-m}\sum_{\bm{k}\in \lbrace 0, 1\rbrace^m}\ket{\bm{k}}\!\bra{\bm{k}} \otimes \rho_E} \\
            &= \norm{\bigotimes_{j=1}^m \rho_{K_jE} - \bigotimes_{j=1}^m 2^{-1} \left(\sum_{k_j \in \lbrace 0, 1 \rbrace}\ket{k_j}\!\bra{k_j} \otimes \rho_{E}\right)} \\
            &= \norm{\begin{aligned}
            &\bigotimes_{j=1}^m \rho_{K_jE} - \bigotimes_{j=1}^{m - 1} 2^{-1} \left(\sum_{k_j \in \lbrace 0, 1 \rbrace}\ket{k_j}\!\bra{k_j} \otimes \rho_{E}\right) \otimes \rho_{K_mE}\\
             &\qquad + \bigotimes_{j=1}^{m - 1} 2^{-1} \left(\sum_{k_j \in \lbrace 0, 1 \rbrace}\ket{k_j}\!\bra{k_j} \otimes \rho_{E}\right) \otimes \rho_{K_mE} - \bigotimes_{j=1}^m 2^{-1} \left(\sum_{k_j \in \lbrace 0, 1 \rbrace}\ket{k_j}\!\bra{k_j} \otimes \rho_{E}\right)
             \end{aligned}} \\
             &\leq \norm{\bigotimes_{j=1}^m \rho_{K_jE} - \bigotimes_{j=1}^{m - 1} 2^{-1} \left(\sum_{k_j \in \lbrace 0, 1 \rbrace}\ket{k_j}\!\bra{k_j} \otimes \rho_{E}\right) \otimes \rho_{K_mE}} \\
            &\quad + \norm{\bigotimes_{j=1}^{m - 1} 2^{-1} \left(\sum_{k_j \in \lbrace 0, 1 \rbrace}\ket{k_j}\!\bra{k_j} \otimes \rho_{E}\right) \otimes \rho_{K_mE} - \bigotimes_{j=1}^{m } 2^{-1} \left(\sum_{k_j \in \lbrace 0, 1 \rbrace}\ket{k_j}\!\bra{k_j} \otimes \rho_{E}\right)},
        \end{aligned}
    \end{equation}
    where the first equality uses the fact that our assumption of memoryless devices give that $\rho_{KE} =\otimes_{j=1}^m \rho_{K_jE}$ and in the last inequality we have used the triangle inequality for trace distances. We can write the last inequality as
    \begin{equation}
        \begin{aligned}
            &\norm{\bigotimes_{j=1}^m \rho_{K_jE} - \bigotimes_{j=1}^{m - 1} 2^{-1} \left(\sum_{k_j \in \lbrace 0, 1 \rbrace}\ket{k_j}\!\bra{k_j} \otimes \rho_{E}\right) \otimes \rho_{K_mE}} \\
    &\quad + \norm{\bigotimes_{j=1}^{m - 1}\left[ 2^{-1} \left(\sum_{k_j \in \lbrace 0, 1 \rbrace}\ket{k_j}\!\bra{k_j} \otimes \rho_{E}\right)\right] \otimes \left[ \rho_{K_mE} - 2^{-1} \left(\sum_{k_j \in \lbrace 0, 1 \rbrace}\ket{k_j}\!\bra{k_j} \otimes \rho_{E}\right)\right]}\\
    &= \norm{\bigotimes_{j=1}^m \rho_{K_jE} - \bigotimes_{j=1}^{m - 1} 2^{-1} \left(\sum_{k_j \in \lbrace 0, 1 \rbrace}\ket{k_j}\!\bra{k_j} \otimes \rho_{E}\right) \otimes \rho_{K_mE}}  + \norm{\rho_{K_mE} - 2^{-1} \left(\sum_{k_j \in \lbrace 0, 1 \rbrace}\ket{k_j}\!\bra{k_j} \otimes \rho_{E}\right)}\\
    &\leq \norm{\bigotimes_{j=1}^m \rho_{K_jE} - \bigotimes_{j=1}^{m - 1} 2^{-1} \left(\sum_{k_j \in \lbrace 0, 1 \rbrace}\ket{k_j}\!\bra{k_j} \otimes \rho_{E}\right) \otimes \rho_{K_mE}} + \mathcal{O}(\sqrt{\epsilon}),
        \end{aligned} 
    \end{equation} 
    where in the final inequality we use Theorem~\ref{thm:theorem2} by noting that the single round key register $K_m$ corresponds to Alice's system $A$ and the $j$th state $\ket{k_j}$ is identified with Alice's post-measurement state $\ket{a_1}$. We also note that in the post-selected final generation rounds, $p(k_j) = \frac{1}{2}$ for $k_j \in \lbrace 0, 1\rbrace$.
    
    By recursively applying this technique $m$ times, we obtain
    \begin{equation}
       % \norm{\otimes_{j=1}^m \rho_{K_jE} - \otimes_{j=1}^m 2^{-1} \left(\sum_{k_j \in \lbrace 0, 1 \rbrace}\ket{k_j}\!\bra{k_j} \otimes \rho_{E}\right)} 
       \norm{\rho_{KE} - 2^{-m}\sum_{\bm{k}\in \lbrace 0, 1\rbrace^m}\ket{\bm{k}}\!\bra{\bm{k}} \otimes \rho_E}
       \leq \sum_{i=1}^m \mathcal{O}(\sqrt{\epsilon}) = \mathcal{O}(m \sqrt{\epsilon}),
    \end{equation}
    where we use $\ket{\bm{k}}$ for $\bm{k} \in \{0,1\}^m$ to denote the pure state associated to the concatenation of $m$ key bits $k_j$ for $j \in \{1, \ldots, m\}$. Finally, we note that $\sum_{\bm{k}\in \lbrace 0, 1\rbrace^m}\ket{\bm{k}}\!\bra{\bm{k}} = \mathds{1}_K$ from completeness, and thus obtain $\norm{\rho_{KE} - 2^{-m}\mathds{1}_{K} \otimes \rho_{E}} \leq \mathcal{O}(m\sqrt{\epsilon})$, which completes the proof.
\end{proof}

%%%%%%%%%%%%%%%%%%%%%%%%%%%%%%%%%%%%%%%%%%%%%%%%%%%%%%%%%%%%%%%%%%%%%%%%%%%%%%%%%%%%%

\section{Protocols for even $N$-cycle NC inequalities}
\label{sec:protocol_even_n}

In this section, we apply our protocol to even $N$-cycle NC inequalities. 
Even $N$-cycle NC inequalities have corresponding exclusivity graphs that are non-cyclic (while their compatibility graphs are cyclic), unlike those discussed in the main text. It is important to note that a robust self-test for these scenarios has not yet been derived using the graph-theoretic approach. Therefore, security for even $N$-cycle only holds for the maximal violation of the corresponding NC inequality (i.e., using Theorem~\ref{thm:theorem1} to derive Theorem~\ref{thm:final_key}). Consequently, the results presented here should be interpreted as an upper bound on the expansion rate using these NC inequalities. Additionally, the protocol for even $N$-cycles requires the extra assumption of rank-$1$ projective measurements, compared to our original protocol.

\begin{figure}
\begin{mdframed}
{\bf Protocol for randomness expansion - even $N$-cycle NC inequalities}
\\
\hrule
\vspace{0.3cm}
\textbf{Parameters and notation:}
\begin{itemize}[leftmargin = 0.3cm]
    \item[] $n \in \mathbb{N}$ - Total number of rounds.

    \item[] $\mathcal{H}$ - Hilbert space of dimension $d \geq 4$.    

    \item[] $N \in \mathbb{N}$ - Total number of measurements. 

    \item[] $\Pi_i$, $i\in \lbrace 1, \ldots, N\rbrace$ - Two-outcome projective measurements with outcomes labeled by $a_i \in \lbrace 0, 1\rbrace$ corresponding to $\mathds{1} - \Pi_i$ and $\Pi_i$ respectively.

    \item[] $\rho \in \mathcal{H}$ - Initial quantum state on which the measurements are performed.

    \item[] $\beta$ - NC inequality corresponding to each scenario with maximum quantum value $\beta_{\text{QT}}$.
    
    \item[] $\left(\rho, \lbrace \Pi_i\rbrace_{i = 1}^{N}\right)$ - Quantum realization that obtains $\beta = \beta_{\text{QT}}$. 

    \item[] $\omega_1 = 1$ and $\omega_0 = \frac{1 + \cos{\frac{\pi}{2 N}}}{3 - \cos{\frac{\pi}{2 N}}}$. 
\end{itemize}
\hrule
\vspace{0.5mm}
\hrule
\vspace{0.3cm}
\textbf{Procedure}
\begin{enumerate}[leftmargin = 0.3cm]
    \item Alice chooses the quantum realization $\left(\rho, \lbrace \Pi_i\rbrace_{i = 1}^{N}\right)$ which achieves $\beta = \beta_{\text{QT}}$ for the NC inequality and sets $j = 1$.

    \item \textbf{While} $j \leq n$:
    \begin{itemize}
        \item[] For $q \in \left[0, 1\right]$, choose $T_j = 0$ with probability $1 - q$ and $T_j = 1$ otherwise. 

        \item[] \textbf{If} $T_j = 0$ (Generation round):
        \begin{itemize}
            \item[] Perform the measurement $\Pi_1$ on state $\rho$ to obtain outcome $a_1$.

            \item[] \textbf{If} $a_1 = 0$:
            \begin{itemize}
                \item[] Record $a_1$ as $k_j$ with probability $\omega_0$.
            \end{itemize}

            \item[] \textbf{Else} $a_1 = 1$:
            \begin{itemize}
                \item[] Record $a_1$ as $k_j$ with probability $\omega_1$.
            \end{itemize}
        \end{itemize}

        \item[] \textbf{Else} $T_j = 1$ (Spot-check round):
        \begin{itemize}
            \item[] Randomly choose $i \in \lbrace 1, \ldots, N\rbrace$ and $l \in \{0, 1\}$ with uniform probability and compute $l' = (i + (-1)^l \mod{N}) + 1$.
        
            \item[] Perform the measurement $\Pi_{i}$ on state $\rho$ to obtain outcome $a_i$. Perform the measurement $\Pi_{l'}$ on the post-measurement state to obtain outcome $a_{l'}$. Record $a_i, a_{l'}, i, l'$.

        \end{itemize}
        \item[] Set $j = j + 1$.
    \end{itemize}

    \item Using the statistics from all spot-check rounds evaluate $\beta$. 

    \item[] \textbf{If} $\beta \neq \beta_{\text{QT}}$:
    \begin{itemize}
        \item[] Abort the protocol.
    \end{itemize}

    \item[] \textbf{Else} 
    \begin{itemize}
        \item[] Obtain the bit string $\textbf{k}$ as a concatenation of all bit values $k_j$.
    \end{itemize}
\end{enumerate}
\end{mdframed}
\caption{The template protocol for randomness expansion using self-testing from even $N$-cycle NC inequalities.}
\label{fig:protocol_local_chsh}
\end{figure}

For even $N$-cycle NC inequalities, Alice can choose between $N$ different dichotomic measurements. The compatibility graph for this scenario is the M\"{o}bius ladder graph of order $2N$~\cite{CDL13,BRX22}, with the NC and contextual bounds given by $\beta_{\text{NCHV}} = N - 1$ and 
$\beta_{\text{QT}} = N\left[1 + \cos{\frac{\pi}{N}}\right]$~\cite{BRX22}, respectively. A non-robust self-test for this NC scenario, along with the probability assignment for each vertex for maximal quantum violation of the inequality, can be found in Ref.~\cite{BRX22}. This probability assignment for each vertex is $p(1|i) = \frac{1}{2}\left[1 + \cos\left({\frac{\pi}{N}}\right)\right]$, where $p(1|i)$ monotonically approaches 0.5 as $N$ increases. Accordingly, we define the post-selection probability in this case (following our original protocol) as $\omega_0 = \frac{1 + \cos{\frac{\pi}{N}}}{3 - \cos{\frac{\pi}{N}}}$.
The full protocol is presented in Fig.~\ref{fig:protocol_local_chsh}, with the corresponding parameters specified here.

Using Eq.~\eqref{eq_app:consumed_rand}, we plot the randomness expansion rate in Fig.~\ref{fig:r_vs_N} as a function of $N$. We find that the rate $r$ tends to $0.9637$ as $N$ increases.

\begin{figure}
    \centering
    \includegraphics[scale = 0.7]{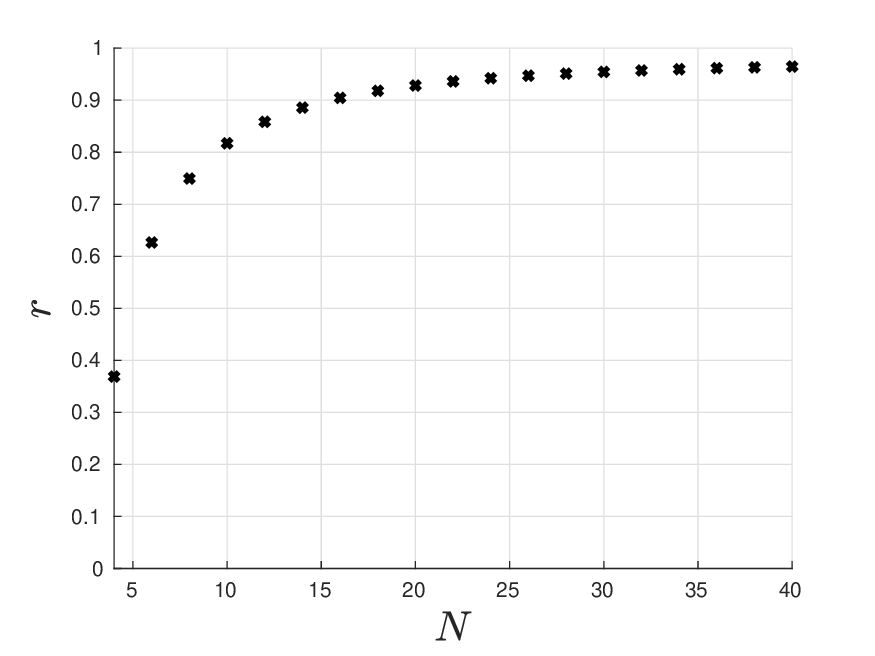}
    \caption{The rate of secure randomness expansion, $r$, as a function the total number of measurements, $N$. We fix the number of experimental rounds as $n = 10^6$ and the probability a round is selected to be a spot-check as $q = 1/\sqrt{n}$.}
    \label{fig:r_vs_N}
\end{figure}

%%%%%%%%%%%%%%%%%%%%%%%%%%%%%%%%%%%%%%%%%%%%%%%%%%%%%%%%%%%%%%%%%%%%%%%%%%%%%%%%%%%%%

\end{document}